\UseRawInputEncoding 

\documentclass[aps,pra,twocolumn,showpacs,superscriptaddress]{revtex4-1}
\usepackage{graphicx}
\usepackage{dcolumn}
\usepackage{bm}
\usepackage{xcolor}
\usepackage{times}
\usepackage{amssymb}
\usepackage{amsmath}
\usepackage{epsfig}
\usepackage{epstopdf}
\usepackage{dsfont}
\usepackage{subfigure}
\usepackage[colorlinks=true,citecolor=blue, linkcolor=blue,urlcolor=blue]{hyperref}
\usepackage{setspace}
\usepackage{bookmark}
\usepackage{amsthm}
\usepackage{mathtools}
\usepackage{amsfonts}

\theoremstyle{definition}
\newtheorem{theorem}{Theorem}
\newtheorem{lemma}{Lemma}
\newtheorem{corollary}{Corollary}
\definecolor{Green3}{rgb}{0.80,0.87,0.76}

\begin{document}

\title{Quantifying and Probing Multipartite Entanglement via Minimum  {Entanglement} Drop}

\author{Dong-Dong Dong}
\affiliation{School of Physics and Optoelectronics Engineering, Anhui University, Hefei
  230601,  People's Republic of China}
\author{Xue-Ke Song}%
\affiliation{School of Physics and Optoelectronics Engineering, Anhui University, Hefei
  230601,  People's Republic of China}

\author{Liu Ye}
\affiliation{School of Physics and Optoelectronics Engineering, Anhui University, Hefei
  230601,  People's Republic of China}
\author{Dong Wang}
\email{dwang@ahu.edu.cn}
\affiliation{School of Physics and Optoelectronics Engineering, Anhui University, Hefei
  230601,  People's Republic of China}

\date{\today}

\begin{abstract}

Quantifying genuine multipartite entanglement remains a significant challenge due to the exponential scaling of computational complexity. In this paper, we propose a novel  {multipartite entanglement monotone} defined by the minimum entanglement drop---the reduction in global one-to-group entanglement when a single constituent particle is traced out.
While analytically rooted in the generalized monogamy inequality, we formulate a computationally efficient variant based on tangle and  negativity to ensure non-vanishing values for W-class states.
We rigorously prove that this quantity  constitutes a valid entanglement monotone under local operations and classical communication. In the tripartite regime, we demonstrate that the minimum tangle drop is physically equivalent to the minimum pairwise concurrence. Furthermore,  { we establish a proof-of-principle operational framework where the entanglement drop serves as a structural heuristic probe: by evaluating the sensitivity of the system to qubit loss, we identify inseparable clusters within certain classes of multipartite states, effectively extracting their connectivity fingerprints, which can uniquely differentiate graph topologies even within the same local Clifford  equivalence class.
 Building on this localized mapping, we highlight its practical utility by integrating it with the classical shadows formalism for efficient experimental estimation, and demonstrate its unique capability to dynamically track the spatiotemporal evolution of entanglement networks. {To further validate its scalability, we derive exact analytical solutions for $n$-qubit W states under environmental noise, revealing the robust scaling behaviors of the proposed measure.} Finally, to ensure a balanced assessment, we candidly acknowledge the fundamental limitations of this heuristic probe, noting that its diagnostic sensitivity strictly vanishes for highly robust, symmetrically correlated states, such as the 5-qubit error-correcting code.}

\end{abstract}

\maketitle

\section{Introduction}
Quantum entanglement is widely recognized as a cornerstone of quantum mechanics and an indispensable resource for a vast array of quantum information processing tasks. These applications range from quantum cryptography \cite{Ekert1991, Gisin2002} and superdense coding \cite{Bennett1992} to quantum teleportation \cite{Bennett1993}, as well as advanced schemes in quantum metrology \cite{Giovannetti2011} and quantum computation \cite{NielsenChuang, Horodecki2009, Amico2008}.
While the characterization of bipartite entanglement is relatively well-established, significant distinctions exist depending on the dimension of the system. For two-qubit systems, entanglement is thoroughly understood, with computable measures such as concurrence for mixed states \cite{Wootters1998} and negativity \cite{Vidal2002}. However, the scenario becomes considerably more intricate for high-dimensional bipartite mixed states. In these cases, extending pure-state measures typically relies on the convex roof construction \cite{PhysRevA.54.3824,PhysRevA.64.042315,PhysRevA.79.012329},
 which requires optimization over all possible ensemble decompositions¡ªa task known to be NP-hard \cite{Gurvits2003}. Although negativity remains computable in higher dimensions, it fails to detect bound entanglement, i.e., entangled states that have a positive partial transpose (PPT) and cannot be distilled \cite{Horodecki1998}.

The quantification of multipartite entanglement presents an even more profound challenge due to the exponentially increasing complexity of the Hilbert space and the rich variety of ways particles can be correlated. For three-qubit systems, it is known that there are two inequivalent classes of genuine multipartite entanglement under stochastic local operations and classical communication (SLOCC): the Greenberger-Horne-Zeilinger (GHZ) class and the W class \cite{Dur2000}. As the system size increases to four qubits, the classification becomes significantly more complex, yielding nine distinct SLOCC families \cite{Verstraete2002}.
In recent years, significant efforts have been made to construct faithful measures of genuine multipartite entanglement (GME).
Ma \emph{et al.} introduced the genuine multipartite concurrence (GMC), defined as the minimum bipartite concurrence across all possible bipartitions, ensuring that the measure vanishes for biseparable states \cite{Ma2011}. More recently, geometric interpretations have provided new insights. Xie and Eberly proposed the ``concurrence fill'' based on the area of a concurrence triangle for three-qubit states \cite{Xie2021}. While initially introduced as a geometric invariant, this measure was subsequently refined to strictly satisfy the monotonicity requirement under local operations, thereby establishing it as a faithful quantifier of entanglement \cite{PhysRevA.107.032405}. In parallel, Li and Shang extended the geometric mean of bipartite concurrences (GBC) to arbitrary multipartite systems, offering a unified framework for high-dimensional entanglement characterization \cite{Li2022}.
 {For a comprehensive progress report on these measures, we refer to a recent review~\cite{MA20252489}. Beyond these geometric schemes, alternative quantifiers have been established to analyze the intricate structure of multipartite correlations, such as the L-entropy measure~\cite{kmpl-mdbx}, as well as GME characterizations in quantum optimization processes~\cite{PhysRevA.111.022434}.}

On the other hand, a fundamental property governing the distribution of entanglement in multipartite systems is monogamy, which dictates that entanglement cannot be freely shared among many parties. First formalized by Coffman, Kundu, and Wootters (CKW) for three qubits \cite{Coffman2000}, the monogamy inequality relates the bipartite entanglement between a focus qubit and the rest of the system (the one-tangle) to the sum of reduced two-qubit entanglements.
The residual entanglement arising from this inequality, known as the three-tangle, serves as a measure of genuine tripartite entanglement.
Scaling these concepts beyond the tripartite regime,
the CKW framework was generalized to $n$-qubit systems by Osborne and Verstraete \cite{Osborne2006}.
To capture the finer correlation structure, Regula \emph{et al.} proposed the ``strong monogamy'' conjecture \cite{Regula2014}, which tightens these constraints by explicitly incorporating the residual entanglement of $m$-particle subsystems ($m < n$) within the global state, rather than limiting the analysis to pairwise sums.
However, despite its significance, this strong monogamy relation remains a conjecture and has not been rigorously proven for general quantum states.

Furthermore, the quantification of GME faces severe computational hurdles.  Even when restricted to pure states, calculating measures like GMC and GBC requires minimizing or averaging over all possible bipartitions of the system. Since the number of bipartitions scales exponentially with the particle number $n$, the computational complexity becomes prohibitive as the system size increases. Consequently, establishing a multipartite entanglement monotone that is both physically meaningful and efficiently computable remains a task of paramount importance.

Motivated by these developments, in this paper, we propose a new  {multipartite entanglement monotone} defined by the minimum entanglement drop induced by qubit loss. We rigorously establish the validity of this quantity by proving its monotonicity under local operations and classical communication (LOCC). To ensure computational efficiency and detect entanglement in W-class states, we further formulate a practical variant based on Negativity, termed the minimum negativity drop. Finally, we demonstrate that this drop quantity serves as a structural probe for identifying inseparable clusters within multipartite systems, effectively revealing the connectivity fingerprints of various graph states,  {demonstrating the striking capability to distinguish distinct network structures even within a single local Clifford (LC) orbit. Furthermore, we showcase the practical advantages of our method by combining it with the classical shadows formalism to bypass exponential measurement costs on near-term devices, and apply it to monitor the dynamic formation of entanglement networks under Ising interactions.  {To benchmark its practical utility beyond small systems, we derive closed-form analytical expressions for $n$-qubit W states subjected to environmental noise, explicitly demonstrating the measure's predictable scaling laws and robustness.} While demonstrating these capabilities, we also explicitly delineate the boundaries of our approach, candidly acknowledging that this heuristic structural probe fails to detect entanglement in highly robust, symmetrically correlated systems, such as the 5-qubit error-correcting code.}

\section{The minimum residual entanglement (MRE) of the generalized CKW inequality}
 {
Before detailing the primary contribution of this work¡ªthe minimum negativity drop¡ªwe first introduce a preliminary concept: the MRE. We explicitly note that the MRE is presented primarily as a theoretical bridge. Formulating the MRE serves two crucial purposes: first, it resolves critical physical inconsistencies found in previous multi-focus monogamy definitions; second, it establishes the core mathematical machinery of this paper, particularly the proof techniques for LOCC monotonicity (Theorem 1) that will be inherited by our subsequent measures. However, we will also demonstrate that the MRE ultimately exhibits a critical deficiency, as it fails to detect certain genuine multipartite entangled states (such as the W state). It is exactly this inherent limitation that naturally motivates the transition to the central framework of our paper in Section III: the entanglement drop induced by particle loss.}

In 2000, Coffman \emph{et al.} first formalized the monogamy of entanglement for a three-qubit system via an inequality, known as CKW inequality \cite{Coffman2000}
\begin{align}
\tau_{A(BC)} \geq \tau_{AB} + \tau_{AC},
  \label{Eq.1}
    \end{align}
where $\tau_{A(BC)}$ represents the square of the   bipartite concurrence (also known as tangle) for the bipartition $A:BC$, which is defined as:
\begin{align}
  \tau_{A(BC)}=2[1-tr(\rho_A^2)]=4det(\rho_A),
  \label{Eq.2}
\end{align}
where $\rho_A$ is the reduced matrix of the subsystem A. The tangle for any mixed state $\rho_{ABC}$ is defined via the convex roof construction
\begin{align}
  \tau_{A(BC)}:=\inf_{\{p_i, |\psi_i\rangle\}} \sum_i p_i \mathcal{\tau}_{A(BC)}(|\psi_i \rangle),
  \label{Eq.3}
\end{align}
where the infimum runs over all pure-state decompositions of $\rho_{ABC}$, $\rho_{ABC}=\sum_i p_i |\psi_i\rangle \langle \psi_i|$.
Additionally, $\tau_{AB}$ and $\tau_{AC}$ denote squared concurrence of reduced states $\rho_{AB}$ and $\rho_{AC}$, respectively.
For a general mixed two-qubit state $\rho_{AB}$, the concurrence $C_{AB}$ is determined by the eigenvalues of the non-Hermitian matrix $R = \rho_{AB} \tilde{\rho}_{AB}$. The formula is given by
\begin{equation}
    C_{AB} = \max\{0, \lambda_1 - \lambda_2 - \lambda_3 - \lambda_4\},
    \label{Eq.c4}
\end{equation}
where $\lambda_1 \geq \lambda_2 \geq \lambda_3 \geq \lambda_4$ are the square roots of the eigenvalues of $R$. Here, $\tilde{\rho}_{AB}$ denotes the spin-flipped density matrix defined as
\begin{equation}
    \tilde{\rho}_{AB} = (\sigma_y \otimes \sigma_y) \rho_{AB}^* (\sigma_y \otimes \sigma_y),
\end{equation}
where $\rho_{AB}^*$ is the complex conjugate of $\rho_{AB}$, and $\sigma_y$ represents the Pauli-$Y$ matrix.

In fact, for an arbitrary $2 \times d$ system, the tangle can be  defined by Eqs. (\ref{Eq.2}) and (\ref{Eq.3}) \cite{PhysRevA.64.042315}.
Subsequently, for an arbitrary $n$-qubit state involving qubits $q_1, q_2, \dots, q_n$, a multipartite monogamy inequality analogous to Eq.~(\ref{Eq.1}) holds~\cite{Osborne2006}:
 {
\begin{equation}
 \tau_{q_1(q_2\dots q_n)} \geq \tau_{q_1q_i} + \tau_{q_1(q_2\dots \hat{q}_i \dots q_n)},
    \label{Eq.9}
\end{equation}}
where $i \in \{2, \dots, n\}$ denotes an arbitrary index of the remaining qubits,  {and the hat notation $\hat{q}_i$ explicitly denotes the omission of the qubit $q_i$ from the sequence.} By iteratively applying this inequality to the remaining high-dimensional subsystem (decomposing it until only single qubits remain), a generalized version involving exclusively two-qubit tangles is obtained:
\begin{align}
\tau_{q_1(q_2...q_n)} \geq \tau_{q_1q_2} + \tau_{q_1q_3}+...+\tau_{q_1q_n}.
\label{Eq.10}
    \end{align}

Notably, for three-qubit pure states, the difference between the left- and right-hand side of Eq. (\ref{Eq.1}) can serve  as a measure of tripartite entanglement, namely residual entanglement, also known as three-tangle
\begin{align}
\tau_{ABC}:=\tau_{A(BC)} - \tau_{AB} - \tau_{AC},
  \label{Eq.6}
    \end{align}
which is invariant under qubit permutation.
For mixed states, the three-tangle is defined via the convex roof construction, taking the infimum over all possible ensembles of pure states.

 {
To generalize this concept to larger multiqubit systems, Ref.~\cite{PhysRevA.71.042331} proposed a multipartite residual entanglement measure by utilizing monogamy inequalities with multiple focus qubits. Specifically, by grouping a subset of $m$ qubits ($1 \le m \le \lfloor n/2 \rfloor$) to act as a collective focus block $S_\alpha$, the residual entanglement with respect to this specific focus is evaluated as:
\begin{equation}
    R_{S_\alpha} = \tau_{S_\alpha | \overline{S}_\alpha} - \sum_{i \in \overline{S}_\alpha} \tau_{S_\alpha | i},
    \label{Eq.MultiFocus_sub}
\end{equation}
where $\overline{S}_\alpha$ denotes the remaining $(n-m)$ qubits. The global multipartite residual entanglement is then defined by taking the minimum over all possible choices of the focus blocks:
\begin{equation}
    R_{\text{global}} = \min_{S_\alpha} \{ R_{S_\alpha} \}.
    \label{Eq.MultiFocus_min}
\end{equation}
}

 {
However, the residual entanglement derived solely from these relations is not regarded as a strictly qualified measure of GME.}
 Generally, a valid GME measure is required to satisfy five postulates \cite{Ma2011}:

(i) it vanishes for biseparable states;

(ii) it is strictly positive for all GME states;

(iii) it satisfies convexity for mixed states;

(iv) it is  invariant under local unitary (LU) operations;

(v) it satisfies monotonicity under LOCC.

\noindent  {Neither the three-tangle nor the generalized multi-focus residual entanglement satisfies all these conditions.} First, the three-tangle crucially fails to satisfy the second condition.
 A prominent counterexample is the $W$ state, defined as
\begin{equation}
    |W\rangle = \frac{1}{\sqrt{3}} (|100\rangle + |010\rangle + |001\rangle).
\end{equation}
While the $W$ state is genuinely entangled, its three-tangle is zero.

 {Second, even within the generalized framework proposed in Ref.~\cite{PhysRevA.71.042331}, the specific instance of residual entanglement defined with a single focus qubit [$m=1$, yielding Eq.~(\ref{Eq.10})] exhibits a critical flaw.} It severely violates condition (i), meaning it does not strictly vanish for biseparable states.
A compelling counterexample is the tensor product of two GHZ states:
\begin{equation}
    |\Psi\rangle = |\text{GHZ}\rangle_{123} \otimes |\text{GHZ}\rangle_{456}.
\end{equation}
For any individual qubit $k$ in this system, the one-to-group squared concurrence is unity ($\tau_{k|\text{rest}}=1$), while all pairwise squared concurrences vanish ($\tau_{kj}=0$). Consequently, the residual entanglement calculated via Eq.~(\ref{Eq.10}) is $1$. However, since $|\Psi\rangle$ is clearly biseparable (separable with respect to the partition $123|456$), a valid GME measure must be zero.
 {Furthermore, this multi-focus approach implicitly relies on high-dimensional monogamy inequalities, }which have subsequently been proven to be invalid for the tangle in generic high-dimensional systems \cite{PhysRevA.75.034305}.

 {
To resolve this deficiency for small systems (three qubits),}  the minimum pairwise concurrence (MPC) was introduced~\cite{Dong2024_MPC}. It is defined as the square root of the sum of the three-tangle and the minimum two-qubit squared concurrence:
\begin{equation}
    \mathcal{M}_{\text{ABC}} = \sqrt{\tau_{ABC} + \min\{\tau_{AB}, \tau_{AC}, \tau_{BC}\}},
\end{equation}
which has been proven to be a qualified multipartite entanglement measure.

 {Motivated by the necessity to overcome these profound theoretical drawbacks---namely,  the failure of single-focus subtractions to vanish for biseparable states and the invalidity of high-dimensional tangle monogamy---we propose a refined, strictly generalized definition of multipartite residual entanglement.}

\begin{theorem}
 For an $n$-qubit pure state, the minimum residual entanglement  with respect to qubit $q_1$ is defined as
\begin{equation}
   {
  \mathcal{R}_{q_1} = \min_{i\in\{2,\dots,n\}} \sqrt{\tau_{q_1(q_2\dots q_n)} - \tau_{q_1q_i} - \tau_{q_1(q_2\dots \hat{q}_i \dots q_n)}}.}
    \label{Eq.ResEntDefinition}
\end{equation}
For mixed states, it is defined via the convex roof extension, taking the infimum over all possible ensembles of pure states.
 This quantity is an entanglement monotone and satisfies four of the five criteria for a GME measure, with the exception of the Condition (ii).
\end{theorem}
\begin{proof}
    We now verify that $\mathcal{R}_{q_1}$ satisfies the required conditions.

    \medskip
    \noindent \textit{Condition (i): Biseparability.}
    Any biseparable state can be represented as a convex combination of biseparable pure states. Consider a pure state $|\psi\rangle$ that is biseparable with respect to a certain partition.

First, consider the case where qubit $q_1$ is fully separable from the rest of the system, i.e., $|\psi\rangle = |q_1\rangle \otimes |\phi\rangle_{\text{rest}}$. Here, the one-to-group tangle $\tau_{q_1(\text{rest})}$ is zero, rendering the $\mathcal{R}_{q_{1}}$ trivially zero.

On the other hand, consider the case where $q_1$ is entangled with a subset of particles $S$, but the total state is separable with respect to a partition $S | \bar{S}$. In this scenario, we can select an index $i$ from the uncorrelated subsystem $\bar{S}$. It follows that $\tau_{q_1 q_i} = 0$. Moreover, since tracing out particle $i$ does not affect the correlations within $S$, we have $\tau_{q_1(\text{rest})} = \tau_{q_1(\text{rest} \setminus q_i)}$. This leads to a zero value for the term corresponding to index $i$ in Eq.~(\ref{Eq.ResEntDefinition}).

Since $\mathcal{R}_{q_1}$ is defined as the minimum over all $i$, $\mathcal{R}_{q_1}(|\psi\rangle) = 0$ holds for any biseparable pure state. Due to the convex roof definition, $\mathcal{R}_{q_1}(\rho) = 0$ for any biseparable mixed state $\rho$.

 \medskip
    \noindent \textit{Condition (iii): Convexity.}
    By definition, for mixed states, $\mathcal{R}_{q_1}$ is constructed via the convex roof.

    \medskip
    \noindent \textit{Condition (iv): LU Invariance.}
    The tangle $\tau$ is known to be invariant under local unitary operations \cite{Coffman2000,PhysRevA.64.042315}. Since the $\mathcal{R}_{q_1}$ for pure states is defined as a minimum of the square root of a linear combination of tangles, it inherits this invariance.

\medskip
\noindent \textit{Condition (v): Monotonicity under LOCC.}
To prove the monotonicity for mixed states, we invoke a necessary and sufficient criterion established by Demkowicz-Dobrza\'nski \textit{et al.} in Ref.~\cite{PhysRevA.74.052303}. It can be formulated as the following lemma:

\begin{lemma}
Let $E$ be a real, nonnegative function defined on pure states, invariant under local unitaries, and satisfying the homogeneity condition $E(a|\psi\rangle) = |a|^2 E(|\psi\rangle)$. Its convex roof extension to mixed states is an entanglement monotone if and only if it satisfies the inequality:
\begin{align}
    E\left( \sqrt{p}|\psi\rangle \otimes |0\rangle_F + \sqrt{1-p}|\phi\rangle \otimes |1\rangle_F \right)  \nonumber \\
    \geq p E(|\psi\rangle) + (1-p) &E(|\phi\rangle),
    \label{Eq.LemmaFLAGS}
\end{align}
where $|\psi\rangle$ and $|\phi\rangle$ are arbitrary $n$-partite pure states, and $|0\rangle_F, |1\rangle_F$ are orthogonal states (flags) of an ancilla attached to any local subsystem.
\end{lemma}
To satisfy the degree-2 homogeneity condition of Lemma 1 (since the tangle $\tau$ is of degree 4), we define the $\mathcal{R}_{q_1}$ using the square root, yielding $\mathcal{R}_{q_1}(a|\psi\rangle) = |a|^2 \mathcal{R}_{q_1}(|\psi\rangle)$.

We now verify that $\mathcal{R}_{q_1}$ satisfies Eq.~(\ref{Eq.LemmaFLAGS}). Consider the superposition state with local flags
\begin{align}
  |\Psi\rangle = \sqrt{p}|\psi\rangle|0\rangle_F + \sqrt{1-p}|\phi\rangle|1\rangle_F.
\end{align}
Due to the orthogonality of the flags ($\langle 0|1\rangle_F=0$), the reduced density matrix of any subsystem $S$ becomes a convex combination: $\rho_S^{\Psi} = p \rho_S^{\psi} + (1-p) \rho_S^{\phi}$. We analyze the terms in the definition of $\mathcal{R}_{q_1}$:

 The one-to-group tangle $\tau_{q_1(\text{rest})}(\rho) \equiv 4\det(\rho_{q_1})$ is a concave function of the density matrix. Thus:
\begin{equation}
    \tau_{q_1(\text{rest})}(\rho_{q_1}^{\Psi}) \geq p \tau_{q_1(\text{rest})}(\rho_{q_1}^{\psi}) + (1-p) \tau_{q_1(\text{rest})}(\rho_{q_1}^{\phi}).
    \label{Eq.ConcavityTau}
\end{equation}

 The subtractive terms $\tau_{q_1 q_i}$ and $\tau_{q_1(\text{rest}\setminus q_i)}$ correspond to mixed-state tangles. Defined via the convex roof, they are inherently convex functions:
\begin{equation}
    \tau_{\text{sub}}(\rho_{\text{sub}}^{\Psi}) \leq p \tau_{\text{sub}}(\rho_{\text{sub}}^{\psi}) + (1-p) \tau_{\text{sub}}(\rho_{\text{sub}}^{\phi}).
    \label{Eq.ConvexitySub}
\end{equation}

Let $\mathcal{D}^{(i)}(|\Psi\rangle)$ denote the term inside the minimization for a specific index $i$. Combining Eq.~(\ref{Eq.ConcavityTau}) and Eq.~(\ref{Eq.ConvexitySub}),  we have:
\begin{align}
    \mathcal{D}^{(i)}(|\Psi\rangle) &= \tau_{q_1(\text{rest})}^{\Psi} - \left( \tau_{q_1 q_i}^{\Psi} + \tau_{q_1(\text{rest}\setminus q_i)}^{\Psi} \right) \nonumber \\
    &\ge p \mathcal{D}^{(i)}(|\psi\rangle) + (1-p) \mathcal{D}^{(i)}(|\phi\rangle).
\end{align}

Finally, the MRE is defined as $\mathcal{R}_{q_1} = \min_i \sqrt{\mathcal{D}^{(i)}}$. Since the square root function $f(x)=\sqrt{x}$ is concave, and the minimum of concave functions is also concave, we obtain:
\begin{align}
    \mathcal{R}_{q_1}(|\Psi\rangle) &= \min_i \sqrt{\mathcal{D}^{(i)}(|\Psi\rangle)} \nonumber \\
    &\geq \min_i \sqrt{p \mathcal{D}^{(i)}(|\psi\rangle) + (1-p) \mathcal{D}^{(i)}(|\phi\rangle)} \nonumber \\
    &\geq \min_i \left[ p \sqrt{\mathcal{D}^{(i)}(|\psi\rangle)} + (1-p) \sqrt{\mathcal{D}^{(i)}(|\phi\rangle)} \right] \nonumber \\
    &\geq p \mathcal{R}_{q_1}(|\psi\rangle) + (1-p) \mathcal{R}_{q_1}(|\phi\rangle).
\end{align}

This confirms that $\mathcal{R}_{q_1}$ satisfies Lemma 1 and is therefore an entanglement monotone.
\end{proof}

One of the counterexamples demonstrating that the MRE does not satisfy Condition (ii) is the $W$ state. For the $W$ state, the MRE vanishes ($\mathcal{R}_{q_1} = 0$), a property similar to that of the 3-tangle.
It is worth noting that, unlike the 3-tangle, the proposed  MRE is not permutation invariant, as it generally depends on the choice of the focus qubit $q_1$.
Nevertheless, one can define a permutation-invariant global  {entanglement monotone} by minimizing over all possible choices of the focus qubit, thus eliminating the dependency on any specific particle.

\section{ENTANGLEMENT DROP INDUCED BY TRACING OUT A PARTICLE}

 {
\subsection{Minimum Tangle Drop}}
To address the issue where the MRE vanishes for the
W state, we propose a new approach based on the reduction of entanglement under particle loss.  {Before formally defining the method, we first clarify the terminology used in this work. We introduce the ¡°entanglement drop¡± as a general conceptual framework, which characterizes the reduction in global one-to-group entanglement when a single constituent particle is traced out. Based on this overarching concept, specific functional implementations can be constructed depending on the chosen bipartite entanglement quantifier. Specifically, when the tangle (squared concurrence) is employed, it yields the ¡°minimum tangle drop¡±; when the negativity is utilized, it results in the ¡°minimum negativity drop¡±. Following this framework, we first introduce the specific definition based on the minimum tangle drop to resolve the
W state deficiency.}

\begin{corollary}
For an arbitrary $n$-qubit pure state, the minimum tangle drop with respect to qubit $q_1$ is defined as:
\begin{equation}
   {
\mathcal{D}_{q_1} = \min_{i\in\{2,\dots,n\}} \sqrt{\tau_{q_1(q_2\dots q_n)} - \tau_{q_1(q_2\dots \hat{q}_i \dots q_n)}}.}
\label{eq:min_ent_drop}
\end{equation}
For mixed states, it is defined via the convex roof extension, taking the infimum over all possible ensembles of pure states.
This quantity is an entanglement monotone and satisfies four of the five criteria for a GME measure, with the exception of Condition (ii).
\end{corollary}

The proof that $\mathcal{D}_{q_1}$ satisfies the required conditions is  analogous to the proof provided for Theorem 1.
This entanglement monotone yields a positive value for the $W$ state because it constitutes the MRE plus a two-qubit entanglement term, and the latter is strictly positive for the $W$ state.
It is worth noting that, specifically for three-qubit systems, minimizing this drop measure over all possible permutations of the qubits yields a global invariant:
\begin{equation}
\mathcal{D}_{\min} = \min_{k \in \{1, 2, 3\}} \mathcal{D}_{q_k}.
\end{equation}
This symmetrized quantity is equivalent to the MPC~\cite{Dong2024_MPC}. Given that the MPC has been rigorously established as a valid measure for verifying GME, this equivalence confirms the physical validity of our proposed measure in the tripartite regime.

Unfortunately, this entanglement monotone is not without limitations and fails to detect certain GME states. A  counterexample is the chain-type graph state, constructed by taking the tensor product of two Bell states and applying a CNOT gate between the intermediate qubits:
\begin{equation}
|\Psi_{\text{chain}}\rangle = \text{CNOT}_{23} \left( |\Phi^+\rangle_{12} \otimes |\Phi^+\rangle_{34} \right).
\end{equation}
Although this state is genuinely multipartite entangled and cannot be separated across any bipartition, its {minimum entanglement drop} vanishes. Specifically, if one traces out a distant particle (e.g., particle 4) relative to the focus qubit $q_1$, the entanglement between $q_1$ and the remaining system remains unchanged compared to the global state. Consequently, the difference term in the minimization vanishes, leading to $\mathcal{D}_{q_1} = 0$.

Furthermore, the practical application of  {Corollary 1} involves certain computational complexities. Specifically, calculating the one-tangle for the reduced density matrix $\rho_{\text{rest}\setminus q_i}$ requires evaluating the entanglement of a mixed state. Although this reduced state---obtained by tracing out a single qubit from a pure state---has a rank of at most 2, which allows the convex roof to be determined through specific analytical methods such as convex characteristic curves \cite{PhysRevA.72.022309,PhysRevA.77.032310}, the process remains mathematically intricate. In contrast, the negativity is based solely on the singular values of the partial transpose, offering a significantly more efficient and direct calculation. To address this, we introduce a variation of the entanglement monotone based on Negativity.

 {
\subsection{Minimum Negativity Drop}}

\begin{theorem}
For an arbitrary $n$-qubit state, the minimum negativity drop with respect to qubit $q_1$ is defined as:
\begin{equation}
\mathcal{D}_{\mathcal{N},q_1} = \min_{i \in \{2, \dots, n\}} \left( \mathcal{N}_{q_1|\mathrm{rest}} - \mathcal{N}_{q_1|\mathrm{rest}\setminus q_i} \right),
\label{eq:negativity_loss}
\end{equation}
where $\mathcal{N}_{A|B} = ||\rho^{T_A}||_1 - 1$ denotes the renormalized negativity across the partition $A|B$ \cite{Vidal2002}. For mixed states, the entanglement monotone is defined via the convex roof extension, taking the infimum over all possible pure-state ensembles. This quantity is an entanglement monotone and satisfies four of the five criteria for a GME measure, with the exception of Condition (ii).
\end{theorem}

\begin{proof}
The proof follows a logic analogous to that of Theorem 1, verifying the required conditions step by step.

\textit{Condition (i): Biseparability.}
Consider a pure state $|\psi\rangle$ that is biseparable with respect to a specific partition.
First, if qubit $q_1$ is separable from the rest of the system ($|\psi\rangle = |q_1\rangle \otimes |\phi\rangle_{\text{rest}}$), the global one-to-group negativity $\mathcal{N}_{q_1|\text{rest}}$ vanishes, leading to $\mathcal{D}_{{\mathcal{N}} ,q_1} = 0$.
Second, consider the case where $q_1$ is entangled within a subsystem $S$, but the total state is a product state with respect to the partition $S|\bar{S}$, i.e., $\rho_{\text{total}} = \rho_S \otimes \rho_{\bar{S}}$. If we trace out a particle $k$ belonging to the uncorrelated subsystem $\bar{S}$, the reduced density matrix becomes $\rho_{\text{rest}\setminus k} = \rho_S \otimes \text{Tr}_k(\rho_{\bar{S}})$.
Calculating the negativity involves the trace norm of the partial transpose. Using the property that the trace norm of a tensor product is the product of the trace norms ($||A \otimes B||_1 = ||A||_1 ||B||_1$), we have:
\begin{align}
||\rho_{\text{total}}^{T_{q_1}}||_1 &= ||(\rho_S \otimes \rho_{\bar{S}})^{T_{q_1}}||_1 \nonumber \\
&= ||\rho_S^{T_{q_1}} \otimes \rho_{\bar{S}}||_1 \nonumber \\
&= ||\rho_S^{T_{q_1}}||_1 \cdot ||\rho_{\bar{S}}||_1.
\label{Eq.23}
\end{align}
Since $\rho_{\bar{S}}$ is a valid density matrix, $||\rho_{\bar{S}}||_1 = \text{Tr}(\rho_{\bar{S}}) = 1$. Consequently, the negativity remains invariant: $\mathcal{N}(\rho_{\text{total}}) = \mathcal{N}(\rho_S)$. Similarly, tracing out a particle from $\bar{S}$ yields a state with the same negativity structure. Thus, the difference term in Eq.~\eqref{eq:negativity_loss} vanishes for any index $i \in \bar{S}$, ensuring $\mathcal{D}_{{\mathcal{N}} ,q_1} = 0$.

   \textit{Condition (iii): Convexity.}
    By definition, for mixed states, $\mathcal{D}_{{\mathcal{N}} ,q_1}$ is constructed via the convex roof.

\textit{Condition (iv): LU Invariance.}
The negativity is defined based on the singular values of the partial transposed density matrix, which are invariant under local unitary transformations \cite{Vidal2002}. Thus, the entanglement monotone inherits this invariance.

\textit{Condition (v): Monotonicity under LOCC.}
We again invoke Lemma 1. A crucial distinction from the tangle (which is homogeneous of degree 4) is that the negativity is linear with respect to the density matrix, $\mathcal{N}(\lambda \rho) = \lambda \mathcal{N}(\rho)$. Consequently, for a state vector $|\psi\rangle$, it satisfies the degree-2 homogeneity condition:
\begin{equation}
\mathcal{N}(a|\psi\rangle) = \mathcal{N}(|a|^2 \rho_\psi) = |a|^2 \mathcal{N}(|\psi\rangle).
\end{equation}
This property allows us to apply Lemma 1 directly without the need for a square root.
Consider the superposition state with orthogonal flags $|\Psi\rangle = \sqrt{p}|\psi\rangle|0\rangle_F + \sqrt{1-p}|\phi\rangle|1\rangle_F$. The reduced density matrix of the system is the convex combination $\rho_S^{\Psi} = p\rho_S^{\psi} + (1-p)\rho_S^{\phi}$.
For a qubit $q_1$ in a pure state, the negativity is functionally equivalent to the concurrence \cite{PhysRevA.79.012329}, making the one-to-group negativity $\mathcal{N}_{q_1|\text{rest}}$ a concave function of the reduced density matrix. Therefore:
\begin{equation}
\mathcal{N}_{q_1|\text{rest}}(|\Psi\rangle) \geq p\mathcal{N}_{q_1|\text{rest}}(|\psi\rangle) + (1-p)\mathcal{N}_{q_1|\text{rest}}(|\phi\rangle).
\end{equation}
Conversely, the subtractive terms $\mathcal{N}_{q_1|\text{rest}\setminus q_i}$ correspond to mixed-state negativities, which are inherently convex functions \cite{Vidal2002}:
\begin{equation}
\mathcal{N}_{q_1|\text{rest}\setminus q_i}^{\Psi} \leq p\mathcal{N}_{q_1|\text{rest}\setminus q_i}^{\psi}+ (1-p)\mathcal{N}_{q_1|\text{rest}\setminus q_i}^{\phi}.
\end{equation}
Combining these inequalities, the function $\mathcal{D}^{(i)}$ inside the minimization satisfies:
\begin{align}
\mathcal{D}^{(i)}(|\Psi\rangle) &= \mathcal{N}_{q_1|\text{rest}}^{\Psi} - \mathcal{N}_{q_1|\text{rest}\setminus q_i}^{\Psi} \nonumber \\
&\geq p \mathcal{D}^{(i)}(|\psi\rangle) + (1-p) \mathcal{D}^{(i)}(|\phi\rangle).
\end{align}
Since the minimum of a set of concave-like functions preserves this property, the final quantity $\mathcal{D}_{{\mathcal{N}} ,q_1}$ satisfies the condition of Lemma 1 and is therefore an entanglement monotone.
\end{proof}

It is worth noting that the limitation regarding Condition (ii) discussed in Theorem 2 vanishes entirely in the tripartite regime.

\begin{corollary}
For three-qubit states, the minimum negativity drop $\mathcal{D}_{\mathcal{N}, q_1}$ satisfies Condition (ii) and is thus a faithful GME measure.
\end{corollary}

\begin{proof}
For a three-qubit pure state $|\psi\rangle_{123}$, the renormalized negativity is exactly equal to the concurrence \cite{PhysRevA.79.012329}, i.e., $\mathcal{N}_{1|23} = C_{1|23}$. For mixed states, the negativity serves as a lower bound to the concurrence, satisfying $\mathcal{N}_{1i} \leq C_{1i}$ \cite{PhysRevLett.95.040504}.

We argue by contradiction. Assume that $|\psi\rangle_{123}$ is a GME state but $\mathcal{D}_{\mathcal{N}, q_1} = 0$.
The condition $\mathcal{D}_{\mathcal{N}, q_1} = 0$ implies that there exists an index $i \in \{2,3\}$ such that the difference vanishes:
\begin{equation}
\mathcal{N}_{1|23} - \mathcal{N}_{1i} = 0 \implies C_{1|23} = \mathcal{N}_{1i}.
\end{equation}
Using the inequality $\mathcal{N}_{1i} \leq C_{1i}$, we have $C_{1|23} = \mathcal{N}_{1i} \leq C_{1i}$. However, from the CKW inequality given in Eq.~\eqref{Eq.1}, we know that $C_{1|23} \ge  C_{1i}$ is always true. Therefore, the inequalities must saturate, yielding the strict equality:
\begin{equation}
C_{1|23} = C_{1i}.
\end{equation}
Substituting this equality back into the CKW relation, we obtain:
\begin{equation}
C_{1j}^2 + \tau_{123} = 0.
\end{equation}
Consequently, the MPC, which is defined based on the combination of the three-tangle and minimum reduced concurrences~\cite{Dong2024_MPC}, must vanish. Since the MPC has been proven to be a faithful GME measure, its vanishing implies that the state $|\psi\rangle_{123}$ is biseparable.

This contradicts the initial assumption that $|\psi\rangle_{123}$ is a GME state. Thus, for any GME state, the strict inequality $\mathcal{N}_{1|23} > \mathcal{N}_{1i}$ must hold, ensuring that $\mathcal{D}_{\mathcal{N}, q_1} > 0$.
\end{proof}

 {
\subsection{Comparison with Existing Multipartite Measures}
\label{sec:comparison}
}
 {
To explicitly position our proposed framework relative to existing multipartite quantities, we analyze the fundamental structural differences between our approach and previous residual entanglement constructions. Specifically, we contrast our minimum drop method with two prominent theoretical frameworks: the generalized multi-focus residual entanglement \cite{PhysRevA.71.042331}, and the pairwise-sum residual negativity \cite{PhysRevA.75.062308}. By applying these measures to specifically constructed biseparable states, we demonstrate how traditional monogamy-based extensions fail to capture genuine multipartite correlations, whereas our approach remains physically consistent.
}

 {
First, we contrast our approach with the generalized multi-focus residual entanglement proposed by Yu and Song \cite{PhysRevA.71.042331}. For a 4-qubit pure state $|\Psi\rangle_{ABCD}$, the specific formalisms diverge significantly. The Yu-Song measure allows grouping multiple qubits into a focus block (e.g., subsystem $AB$), relying on the assumed high-dimensional monogamy relation:
\begin{equation}
    C^2_{AB|CD} \geq C^2_{AB|C} + C^2_{AB|D},
    \label{eq:yusong_ineq}
\end{equation}
where $C^2(\rho) = 2[1 - \text{Tr}(\rho^2)]$ represents the bipartite squared concurrence extended via the convex roof. In contrast, our proposed measure $\mathcal{D}_{q_1}$ (or $\mathcal{D}_{\mathcal{N}, q_1}$) strictly utilizes a single-qubit focus (e.g., $A$) and evaluates the one-to-group entanglement drop upon particle loss:
\begin{equation}
    \mathcal{D}_{A} = \min_{x \in \{B,C,D\}} \left( \mathcal{N}_{A|rest} - \mathcal{N}_{A|rest\setminus x} \right).
    \label{Eq33a}
\end{equation}
}

 {
The conceptual danger of utilizing a multi-qubit focus block is that the assumed monogamy inequality in Eq.~(\ref{eq:yusong_ineq}) is mathematically invalid for generic quantum states, leading to physically ill-defined negative residual entanglement. Consider a 4-qubit state consisting of two independent Bell pairs:
\begin{equation}
    |\Psi\rangle_{ABCD} = |\Phi^+\rangle_{AC} \otimes |\Phi^+\rangle_{BD},
\end{equation}
where $|\Phi^+\rangle = \frac{1}{\sqrt{2}}(|00\rangle + |11\rangle)$. Let us evaluate the terms in Eq.~(\ref{eq:yusong_ineq}). For the global bipartite entanglement $C^2_{AB|CD}$, tracing out $C$ and $D$ yields a completely mixed reduced density matrix $\rho_{AB} = \frac{I_2}{2}_A \otimes \frac{I_2}{2}_B = \frac{I_4}{4}$. Its squared concurrence is 1.5.
To calculate the subsystem entanglement $C^2_{AB|C}$, we trace out $D$, yielding $\rho_{ABC} = |\Phi^+\rangle\langle\Phi^+|_{AC} \otimes \frac{I_2}{2}_B$. Since this is a product state between $AC$ and $B$, the entanglement between the block $AB$ and $C$ is entirely contributed by the pure Bell pair $AC$. By the convex roof construction, $C^2_{AB|C} = 1$. By symmetry, tracing out $C$ yields $C^2_{AB|D} = 1$. Substituting these into Eq. (\ref{eq:yusong_ineq}) yields $ 1.5 \geq 2$.
This striking contradiction demonstrates that the high-dimensional monogamy inequality is severely violated, which would yield a negative (unphysical) residual entanglement. }

 {
Second, our structural framework resolves the intrinsic failures of negativity-based residual entanglement defined via pairwise sums. Ou and Fan extended the CKW inequality to $n$-qubit negativity using a single focus qubit $q_1$, defining the residual negativity as \cite{PhysRevA.75.062308}:
\begin{equation}
    \pi_{q_1} = \mathcal{N}_{q_1|rest}^2 - \sum_{i \neq q_1} \mathcal{N}_{q_1|i}^2.
\end{equation}
While this avoids the multi-focus paradox, it fails drastically to vanish for biseparable states. Even if one attempts to construct a global invariant by minimizing $\pi_{q_1}$ over all possible focus qubits, the deficiency persists. Consider the 6-qubit pure biseparable state $|\Psi\rangle = |\text{GHZ}\rangle_{123} \otimes |\text{GHZ}\rangle_{456}$ discussed in Sec.~II. For any arbitrarily chosen focus qubit $k \in \{1,\dots,6\}$, its global one-to-group negativity is $\mathcal{N}_{k|rest} = 1$. Since tracing out the remaining qubits to obtain any two-qubit reduced density matrix yields a completely separable mixed state, all pairwise negativities strictly vanish ($\mathcal{N}_{k|j} = 0$ for all $j \neq k$). Consequently, the proposed residual negativity yields $\pi_k = 1 - 0 = 1$ for every single qubit. Minimizing over all foci still yields a global value of $1$, falsely indicating genuine 6-partite entanglement for a strictly biseparable state.
}

 {
In sharp contrast, our proposed method elegantly avoids this breakdown. Because we restrict the focus to a single qubit (e.g., $q_1 = A$), the entanglement is rigorously bounded by the one-to-group monogamy relations. Applying our minimum negativity drop $\mathcal{D}_{\mathcal{N}, A}$ to the same state $|\Psi\rangle_{ABCD}$, the global negativity is $\mathcal{N}_{A|BCD} = \mathcal{N}_{A|C} = 1$. If we trace out $B$ or $D$ (which are uncorrelated with $A$), the local entanglement structure remains intact, yielding $\mathcal{N}_{A|CD} = \mathcal{N}_{A|BC} = 1$, corresponding to an negativity drop of $1 - 1 = 0$. Consequently, the minimum negativity drop $\mathcal{D}_{\mathcal{N}, A} = 0$. This result is perfectly consistent with physical reality: since the state is biseparable across the $AC|BD$ partition, it contains no genuine 4-partite entanglement, and our measure correctly vanishes.}

 {
\section{Detection of Inseparable Clusters via Minimum Negativity Drop}
}

While the minimum entanglement drop established in the previous sections may vanish for certain classes of genuine multipartite entangled states (such as the chain-type graph state discussed in Corollary 1), the phenomenon of entanglement drop induced by qubit loss possesses intrinsic physical significance. In this section, we demonstrate that this drop serves as a witness for identifying the inseparable structural components within a multipartite system based on negativity.

\begin{theorem}
For an arbitrary $n$-qubit pure state, if the one-to-group negativity of $q_1$ decreases upon tracing out a specific particle $x$ (where $x \neq q_1$), i.e.,
\begin{equation}
\mathcal{N}_{q_1|\mathrm{rest}} > \mathcal{N}_{q_1|\mathrm{rest}\setminus x}, \label{eq:31}
\end{equation}
then $q_1$ and $x$ belong to a mutually inseparable subset. We define the ``$k$-partite directly entangled cluster'' $S_{q_1}$ as the set containing $q_1$ and all such particles $x$, where $k = |S_{q_1}|$. The multipartite entanglement within this cluster is quantified by the restricted minimum entanglement drop:
\begin{equation}
\mathcal{D}_{S_{q_1}} \coloneqq \min_{x \in S_{q_1} \setminus \{q_1\}} \left( \mathcal{N}_{q_1|\mathrm{rest}} - \mathcal{N}_{q_1|\mathrm{rest}\setminus x} \right).
\end{equation}
Similar to Theorem 2, the generalization to mixed states follows the convex roof construction, inheriting the established properties.
\end{theorem}

\begin{proof}
We proceed by contraposition. Suppose that $q_1$ and $x$ are separable. This implies the existence of a bipartition $M|\bar{M}$ with $q_1 \in M$ and $x \in \bar{M}$ such that the global state factorizes as $|\Psi\rangle = |\phi\rangle_M \otimes |\chi\rangle_{\bar{M}}$. As established in the proof of Theorem 2 [see Eq.~(\ref{Eq.23})], tracing out a particle from an uncorrelated subsystem does not alter the singular values of the partial transpose, thereby leaving the negativity of $q_1$ invariant. Consequently, the observation of a strictly positive entanglement drop, $\mathcal{N}_{q_1|\mathrm{rest}} > \mathcal{N}_{q_1|\mathrm{rest} \setminus x}$, negates the assumption of separability, confirming that $q_1$ and $x$ must belong to a mutually inseparable component.
\end{proof}

 {
At a fundamental level, the physical mechanism behind this entanglement drop is intimately connected to the monogamy properties of quantum correlations, specifically the Koashi-Winter relation \cite{PhysRevA.69.022309}. Using the tangle-based formulation \cite{Osborne2006}, this monogamy relation for a tripartite pure state $\rho_{ABC}$ is expressed as:
\begin{equation}
     \tau(\rho_{A|BC}) = I_2^\leftarrow(\rho_{AB}) + \tau(\rho_{AC}),
\end{equation}
where $I_2^\leftarrow(\rho_{AB})$ denotes the classical correlation within subsystem $AB$.  From this relation, one can explicitly see that the tangle drop $\tau(\rho_{A|BC}) - \tau(\rho_{AC})$ vanishes if and only if the classical correlation $I_2^\leftarrow(\rho_{AB})$ is strictly zero, which dictates that the subsystem $AB$ is a completely uncorrelated product state ($\rho_{AB} = \rho_A \otimes \rho_B$).
}

 {
Crucially, our proposed negativity drop possesses the exact same diagnostic capability as the tangle drop in identifying the separability of qubits. This operational equivalence can be rigorously demonstrated as follows. Let $\mathcal{C}$ denote the concurrence (the square root of the tangle, $\mathcal{C} = \sqrt{\tau}$) and $\mathcal{N}$ denote the negativity. For any pure state bipartition where one subsystem is a single qubit (e.g., the partition $A|BC$), the concurrence and the negativity are exactly equal \cite{PhysRevA.79.012329}, yielding $\mathcal{N}_{A|BC} = \mathcal{C}_{A|BC}$. However, for an arbitrary mixed state such as the reduced density matrix $\rho_{AC}$, the negativity constitutes a strict lower bound for the concurrence \cite{PhysRevLett.95.040504}, meaning $\mathcal{C}_{AC} \ge \mathcal{N}_{AC}$.
}
 {
Suppose the negativity drop vanishes, yielding $\mathcal{N}_{A|BC} = \mathcal{N}_{AC}$. Substituting this into the previous relations, we obtain:
\begin{equation}
    \mathcal{C}_{AC} \ge \mathcal{N}_{AC} = \mathcal{N}_{A|BC} = \mathcal{C}_{A|BC}.
\end{equation}
On the other hand, tracing out particle $B$ from the system $ABC$ to obtain $\rho_{AC}$ is a local operation. Since concurrence is an entanglement monotone, it is non-increasing under partial trace \cite{PhysRevA.74.052303}, which requires $\mathcal{C}_{AC} \le \mathcal{C}_{A|BC}$. For both inequalities to hold simultaneously, we must have strictly $\mathcal{C}_{A|BC} = \mathcal{C}_{AC}$. Consequently, a zero negativity drop strictly necessitates a zero tangle drop.
}

 {
Furthermore, this mathematical relationship reveals a compelling physical feature: the negativity drop generally exhibits a more pronounced response to particle loss than the concurrence drop. By directly subtracting the mixed-state inequality ($\mathcal{N}_{AC} \le \mathcal{C}_{AC}$) from the pure-state equality ($\mathcal{N}_{A|BC} = \mathcal{C}_{A|BC}$), we inherently obtain:
\begin{equation}
    \mathcal{N}_{A|BC} - \mathcal{N}_{AC} \ge \mathcal{C}_{A|BC} - \mathcal{C}_{AC}.
\end{equation}
This indicates that the negativity drop is strictly greater than or equal to the concurrence drop. Consequently, using the minimum negativity drop not only preserves the precise diagnostic equivalence for biseparability (the "zero" points), but also serves as a more sensitive and amplified indicator for detecting structural entanglement loss in generic mixed states.
}

 {
Given this robust theoretical equivalence, one might conceive a straightforward step-by-step probing method to identify inseparable clusters:
(1) First, trace out all particles except the focus qubit $q_1$ and a specific target particle $q_k$ to obtain the two-qubit reduced density matrix $\rho_{q_1 q_k}$.
(2) Second, verify whether $\rho_{q_1 q_k}$ is a product state by checking if $\rho_{q_1 q_k} = \rho_{q_1} \otimes \rho_{q_k}$.
(3) Finally, assign $q_1$ and $q_k$ to the same inseparable cluster if they are correlated ($\rho_{q_1 q_k} \neq \rho_{q_1} \otimes \rho_{q_k}$).
}

 {
While this localized procedure intuitively identifies pairwise correlations, it serves primarily as a qualitative verification.  The primary advantage of our proposed method lies in providing a systematic, quantitative, and computationally efficient operational framework. By formalizing the structural probe via the minimum negativity drop, our approach bypasses complex optimizations and systematically quantifies the exact global entanglement supported by specific particles, effectively extracting the connectivity fingerprints of complex multipartite systems.
}

 {
Building upon this quantitative framework, Theorem 3 translates these theoretical advantages into a powerful operational tool for analyzing the underlying entanglement structure of arbitrary quantum states. Specifically,}  by evaluating the entanglement drop of a focus particle $q_1$ against all other particles, we can identify a cluster $S_{q_1}$. Depending on the size of this cluster, we distinguish two cases:

\begin{itemize}
    \item Case I: $|S_{q_1}| = n$. The state is genuinely multipartite entangled, as exemplified by GHZ states or W states.

    \item Case II: $|S_{q_1}| < n$. The system contains at least a $k$-partite inseparable block (note that this identification does not strictly imply separability from the remaining particles outside the block).
\end{itemize}

Furthermore, this classification facilitates a recursive decomposition strategy to simplify GME verification. One can first identify the inseparable block $S_{q_1}$, then select a particle $p \notin S_{q_1}$ to determine its respective block $S_p$, and repeat this process until all particles are assigned to disjoint inseparable clusters $C_1, C_2, \ldots, C_m$. The problem of verifying GME then reduces to checking the separability between these composite blocks, which may significantly reduce the computational complexity compared to optimizing over exponentially many microscopic bipartitions.

 {
It is crucial to emphasize that Theorem 3 serves primarily as a heuristic, proof-of-principle tool for inferring the underlying connectivity structure of quantum states. To systematically evaluate its domain of applicability, we benchmark our method against the LC equivalence classes of small graph states. }

\begin{figure}[tpb]
\centering
\includegraphics[width=0.9\columnwidth]{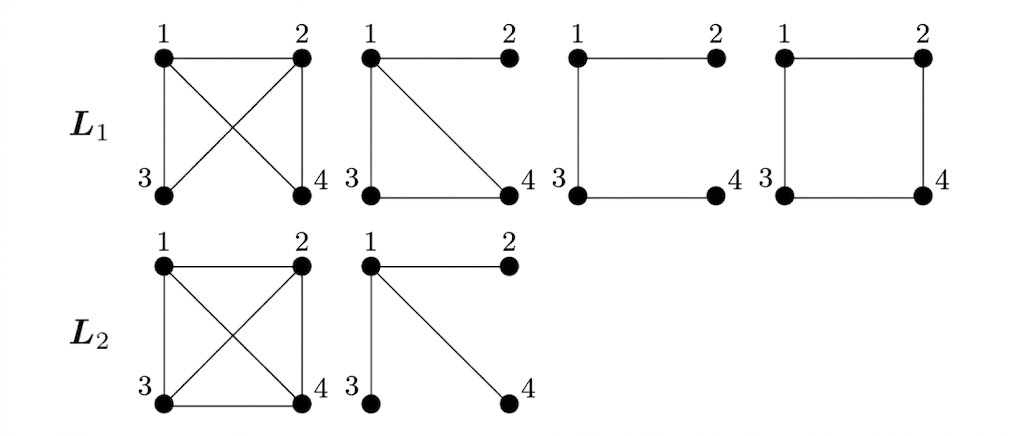} %
\caption{ {The two LC orbits for 4-qubit connected graph states, with nodes explicitly labeled 1 to 4. The top row represents the orbit $L_1$, and the bottom row represents the orbit $L_2$. States within the same orbit exhibit different local entanglement sensitivities, yielding distinct cluster fingerprints under our structural probe.}}
\label{fig:LC_orbits}
\end{figure}

 {
For $n=3$ systems, all connected graph states belong to a single LC equivalence class (locally equivalent to the GHZ state). Due to their global dependency, removing any single particle destroys the entanglement relative to the focus qubit. Consequently, the method reliably identifies a single global cluster containing all particles, $\mathcal{S} = \{1, 2, 3\}$.}

 {
As the system scales to $n=4$ qubits, the landscape of connected graph states partitions into exactly two distinct LC orbits, denoted as $L_1$ and $L_2$ (see Fig.~\ref{fig:LC_orbits}) \cite{danielsen2005selfdualquantumcodesgraphs}. Our cluster identification method fundamentally distinguishes these two orbits: states in $L_1$ consistently partition into two disjoint two-qubit clusters, while states in $L_2$ form a single four-qubit cluster. Interestingly, while states within the same orbit share equivalent global entanglement properties, our method reveals that their local vulnerability to particle loss¡ªand thus their specific structural fingerprint¡ªvaries significantly depending on the specific graph topology. Based on the explicit node numbering in Fig.~\ref{fig:LC_orbits}, we observe the following:
\begin{itemize}
    \item \textbf{Orbit $\boldsymbol{L_1}$:} All four graph topologies in this orbit are characterized by a fragmented structure consisting of two inseparable pairs. However, the exact distribution of particles within these clusters differs. For the first three graphs, the inseparable clusters are consistently identified as the local pairs $\{1, 2\}$ and $\{3, 4\}$, reflecting that intra-pair correlations dominate their structural robustness. Conversely, for the fourth graph (the square ring), the entanglement monotone surprisingly identifies the non-adjacent diagonal pairs as the inseparable clusters, yielding $\{1, 4\}$ and $\{2, 3\}$. This result reveals a distinct sensitivity to long-range interference effects over nearest-neighbor interactions.
    \item \textbf{Orbit $\boldsymbol{L_2}$:} For both the fully connected graph and the star graph in orbit $L_2$, the global dependency is much stronger. Removing any specific particle induces a global entanglement drop relative to the focus qubit. Consequently, the method identifies a single global cluster containing all particles, $\mathcal{S} = \{1, 2, 3, 4\}$, reflecting a GHZ-type macro-robustness.
\end{itemize}}

 {
This systematic analysis candidly highlights a fundamental characteristic of our framework: the identified clusters are generally not invariant under LC operations. Rather than serving as an abstract GME classifier, the method acts as a physical probe sensitive to the explicit hardware-level routing of entanglement and local connectivity.} Consequently, identifying these localized clusters $S_k$ constitutes only the preliminary step in characterizing the global state. For a complete verification of the genuine multipartite entanglement structure, a subsequent examination is required: one must determine whether these identified composite blocks are separable from one another.

 {
Recognizing this heuristic nature,} it is important to acknowledge the limitations of this method. A particularly challenging counterexample is the five-qubit error-correcting code state~\cite{PhysRevLett.77.198, PhysRevA.54.3824}, also known as the 5-qubit stabilizer state~\cite{PhysRevA.54.1862}. This state exhibits a remarkable property: while the global state is a highly entangled pure state, the reduced density matrix of any subsystem involving fewer than 3 qubits appears completely random (i.e., maximally mixed).
Physically, this implies that for any pair of particles $q_1$ and $x$, their reduced state factorizes as $\rho_{q_1 x} = \rho_{q_1} \otimes \rho_{x}$. Since there are no pairwise correlations (classical or quantum), the removal of particle $x$ does not affect the local entanglement structure of $q_1$ relative to the rest of the system. Consequently, the direct entanglement drop vanishes ($\mathcal{N}_{q_1|\mathrm{rest}} - \mathcal{N}_{q_1|\mathrm{rest}\setminus x} = 0$), thereby highlighting the state's exceptional robustness against qubit loss.

  {
\section{Operational Advantages and Practical Applications}}
\label{sec:computability}

 {
  \subsection{Computational Scalability and Experimental Estimation}
To comprehensively address the practical utility of the minimum negativity drop, it is crucial to benchmark its computational complexity against existing global measures and discuss its experimental feasibility on near-term devices.
}

 {
The evaluation of GME via traditional measures, such as the GMC or the generalized geometric measure (GGM) \cite{PhysRevA.81.012308}, inherently suffers from the exponential scaling of the Hilbert space. Specifically, verifying GME via GMC requires minimizing the bipartite concurrence over all $2^{n-1}-1$ possible bipartitions. For highly symmetric states (e.g., standard GHZ or W states), this calculation can be analytically simplified. However, a comprehensive algebraic geometry analysis by Sauerwein \textit{et al.}  rigorously demonstrated that for $n>4$ qubits, generic pure states possess a ``trivial stabilizer'' \cite{PhysRevX.8.031020}. This means they completely lack local symmetries. Consequently, for generic multi-qubit states, no algebraic shortcuts exist to bypass the bipartition loop: evaluating GMC or GGM strictly demands executing the full $\mathcal{O}(2^n)$ optimization, severely limiting their scalability.}

 {
In stark contrast, our proposed operational framework fundamentally bypasses this global exponential search by deploying a localized structural probe. Instead of searching all bipartitions, our method identifies inseparable clusters by tracing out individual particles, shifting the focus from global optimization to local connectivity mapping.}

 {
To illustrate this computational advantage, consider a composite pure state consisting of an $m$-qubit GHZ state and an $(n-m)$-qubit W state: $|\Psi\rangle = |\text{GHZ}\rangle_m \otimes |W\rangle_{n-m}$.
Using the GMC, without any prior knowledge of the state's internal structure, identifying that this state is biseparable (i.e., $\text{GMC}=0$) requires evaluating the concurrence across all $2^{n-1}-1$ bipartitions.
Using our minimum drop framework combined with the cluster-finding algorithm, the complexity is drastically reduced:
(1) Selecting a focus qubit $q_1$ in the $m$-partite subset, tracing out the other particles sequentially requires $n-1$ pairwise checks to accurately identify the $m$-qubit inseparable cluster.
(2) Selecting another focus qubit $q_k$ in the remaining subset requires $n-m-1$ checks to map the $(n-m)$-qubit cluster.
(3) A final check between the two identified clusters confirms their separability.
This approach effectively reduces the intractably large $\mathcal{O}(2^n)$ scaling inherent to global bipartition measures by converting GME detection into a highly scalable structural problem.
Including the initial cluster identification phase, the total computational complexity is bounded by $\mathcal{O}(n^2) + \mathcal{O}(2^k)$, where the polynomial term $\mathcal{O}(n^2)$ corresponds to mapping the local correlations, and the exponential term $\mathcal{O}(2^k)$ represents the final optimization over the $k$ identified macroscopic clusters ($k \le n$).
 {It is crucial to acknowledge the worst-case scenario where $k$ can potentially reach $n$, such as the 5-qubit error-correcting code state, where the exponential scaling $\mathcal{O}(2^n)$ remains unavoidable. Nevertheless, for many physically relevant states, our method remains highly effective. For instance, standard GHZ and W states are directly identified as a single global cluster in their entirety, reducing the complexity to strictly $\mathcal{O}(n)$. For other generic states, the final computational complexity dynamically depends on their specific underlying structural topologies.}
For the specific guided example above  {featuring the composite GHZ and W state}, since the algorithm efficiently groups the system into merely $k = 2$ massive clusters, the exponential bottleneck is entirely bypassed, and the entire verification procedure operates strictly in linear time, $\mathcal{O}(n)$.
}

 {
Beyond theoretical computability, the structural probe enabled by our framework is exceptionally well-suited for experimental implementation on Noisy Intermediate-Scale Quantum (NISQ) devices. Calculating the negativity or concurrence via full quantum state tomography requires $\mathcal{O}(d^{2n})$ measurements, which is experimentally prohibitive \cite{NielsenChuang}. }

 {
However, as mathematically established in Sec.~IV, the vanishing of the negativity drop is strictly equivalent to whether the reduced two-qubit density matrix $\rho_{q_1 q_k}$ factorizes into a product state. For a globally pure multi-qubit state, determining whether $\rho_{q_1 q_k} = \rho_{q_1} \otimes \rho_{q_k}$ can be elegantly translated into a purity product check:
\begin{equation}
    \text{Tr}(\rho_{q_1 q_k}^2) = \text{Tr}(\rho_{q_1}^2) \cdot \text{Tr}(\rho_{q_k}^2).
\end{equation}
If this strict equality is violated, the qubits $q_1$ and $q_k$ are correlated and belong to the same inseparable cluster.}

 {
Crucially, these purity checks can be directly and efficiently estimated using randomized measurements and the classical shadows formalism \cite{PhysRevLett.125.200501}. The core conclusion of Ref.~\cite{PhysRevLett.125.200501} establishes that the purities $\text{Tr}(\rho^2)$ coincide with the second moment of the partially transposed density matrix ($p_2$), and bipartite entanglement in mixed states can be rigorously certified by evaluating the third moment ($p_3$) via the condition $p_3 < p_2^2$.
Integrating this powerful tool into our framework yields a complete, experimentally viable protocol: after utilizing $p_2$ (purity product checks) to map out all inseparable clusters, we can subsequently apply the $p_3 < p_2^2$ criteria to evaluate the bipartitions between these identified macroscopic clusters. If the condition holds across the cluster bipartitions, genuine multipartite entanglement is experimentally confirmed.
}

 {
In this context, it is essential to distinguish between the experimental measurement cost and the classical post-processing complexity. As emphasized in Ref.~\cite{PhysRevLett.125.200501}, a remarkable feature of this randomized measurement protocol is that it only requires single-qubit control and allows for the estimation of many distinct PT moments from the exact same experimental data.
In particular, arbitrary orders $n \ge 2$ and arbitrary partitions $A, B$ (both connected and disconnected) can be simultaneously evaluated.
Consequently, whether one evaluates all $2^{n-1}-1$ bipartitions (as required by traditional measures like GMC) or executes our full two-step algorithmic procedure---performing a polynomial number of localized cluster identifications followed by evaluating the entanglement strictly between the resulting macroscopic clusters---the number of actual quantum experimental runs remains invariant.
Instead, the true bottleneck for traditional multipartite measures lies entirely in the classical hardware: evaluating $\mathcal{O}(2^n)$ non-linear functionals (such as $p_2$ and $p_3$) from the shadow representations is an exponentially hard data-processing task that quickly exhausts classical memory and runtime. Our structural framework fundamentally resolves this by compressing the classical post-processing complexity  {for certain modularly structured states} from an intractable $\mathcal{O}(2^n)$ down to a highly efficient $\mathcal{O}(n^2) + \mathcal{O}(2^k)$  {($k \le n$)}, maximizing the utility of near-term quantum devices without being bottlenecked by classical computational limits.
}

 {
  \subsection{Revealing Structural Evolution in Quantum Dynamics}
Given the growing interest in the role of GME in quantum dynamics and the execution of near-term quantum algorithms (e.g., variational quantum eigensolvers and quantum approximate optimization algorithms), tracking GME has become a key benchmark \cite{PhysRevA.111.022434}. Standard measures such as the GGM or GMC are widely employed to quantify the overall multipartite correlations. However, because these traditional measures compress the rich complexity of the quantum state into a single global scalar, they often obscure the underlying topological structure of the entanglement network as it dynamically forms. }

 {
To demonstrate that our proposed negativity drop provides profound physical insights beyond what is offered by existing global measures, we analyze the dynamical generation of 4-qubit graph states. Consider a system initialized in the fully separable state $|\psi(0)\rangle = |+\rangle^{\otimes 4}$, evolving under the standard Ising-type interaction Hamiltonian $H = \sum_{(i,j)\in E} \sigma_z^{(i)} \sigma_z^{(j)}$ commonly used in quantum optimization and graph state generation. We compare two distinct interaction topologies: the Line graph (with edges $E_{\text{line}} = \{(1,2), (2,3), (3,4)\}$) and the Star graph (with edges $E_{\text{star}} = \{(1,2), (1,3), (1,4)\}$). The system evolves via $U(t) = \exp(-itH)$, yielding the respective ideal graph states at $t = \pi/4$.}

\begin{figure}[htp]
    \centering
    \includegraphics[width=0.9\linewidth]{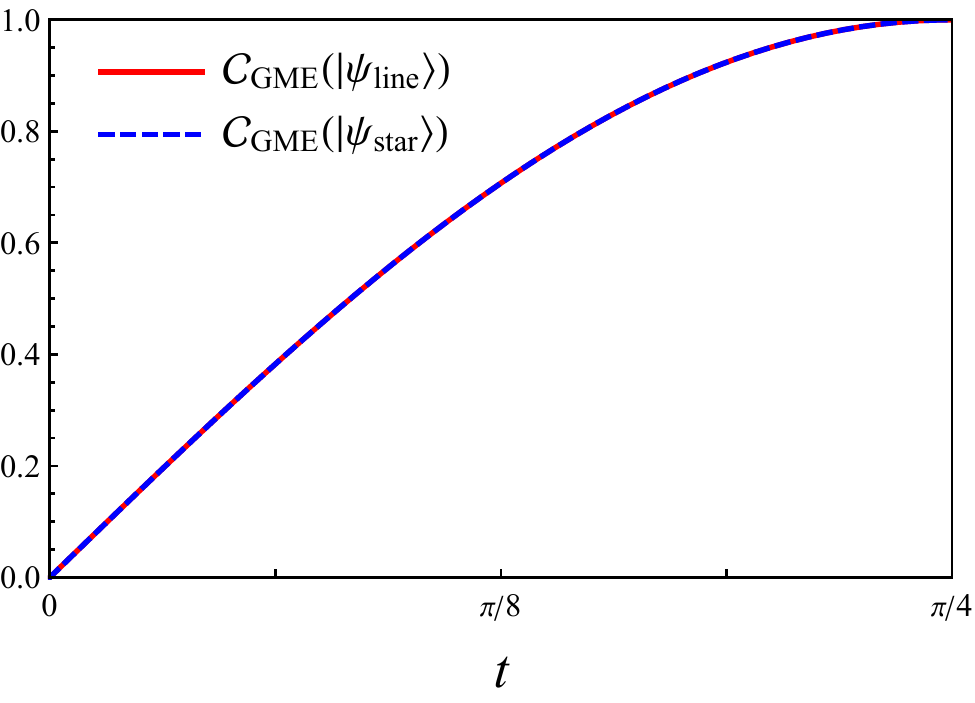}
    \caption{ {The dynamical evolution of the GMC for the 4-qubit Line graph state ($\mathcal{C}_{\text{GME}}(|\psi_{\text{line}}\rangle)$, red solid curve) and the Star graph state ($\mathcal{C}_{\text{GME}}(|\psi_{\text{star}}\rangle)$, blue dashed curve) under Ising interactions. The global measures for both distinct topologies are virtually indistinguishable, monotonically increasing to their maximum values.}}
    \label{fig:global_dynamics}
\end{figure}

 {
As illustrated in Fig.~\ref{fig:global_dynamics}, the evolution of the GMC exhibits nearly identical trajectories for both the Line and Star topologies. As a standard global measure, the GMC is defined by minimizing the bipartite concurrence across all possible bipartitions $\gamma$ of the system:
\begin{equation}
    \mathcal{C}_{\text{GME}}(|\psi\rangle) = \min_{\gamma} \sqrt{2\left[1-\text{Tr}(\rho_{\gamma}^2)\right]}.
\end{equation}
An observer relying solely on such a conventional scalar measure would conclude that the entanglement dynamics of the two systems are fundamentally equivalent, completely failing to detect the distinct graph topologies being constructed at the hardware level.}

\begin{figure}[htp]
    \centering
    \includegraphics[width=0.9\linewidth]{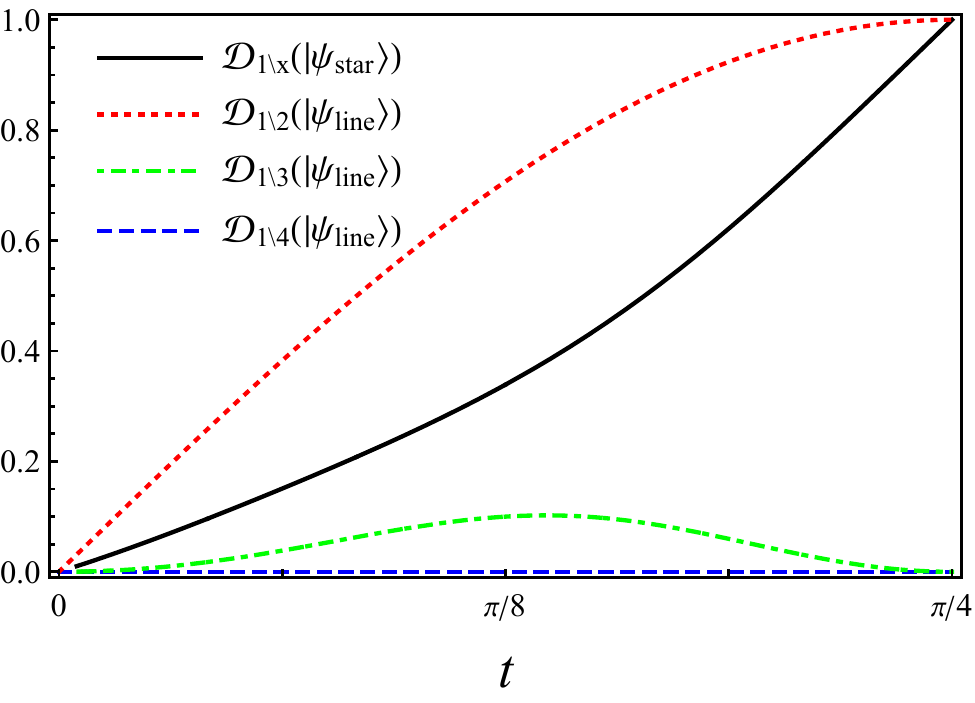}
    \caption{ {The dynamic structural fingerprint revealed by the individual negativity drops $\mathcal{D}_{1\setminus x}(t)$ with respect to focus qubit 1. The black solid curve represents the uniform drop $\mathcal{D}_{1\setminus x}(|\psi_{\text{star}}\rangle)$ for all peripheral nodes ($x \in \{2,3,4\}$) in the Star state. For the Line state, the specific drops $\mathcal{D}_{1\setminus 2}(|\psi_{\text{line}}\rangle)$ (red dotted), $\mathcal{D}_{1\setminus 3}(|\psi_{\text{line}}\rangle)$ (green dashed), and $\mathcal{D}_{1\setminus 4}(|\psi_{\text{line}}\rangle)$ (blue dashed) map the nearest, next-nearest, and distant node dependencies, explicitly capturing the spatiotemporal spread of correlations.}}
    \label{fig:drop_dynamics}
\end{figure}

 {
In stark contrast, our structural probe uniquely unveils the rich spatiotemporal dynamics of the correlation network. By selecting qubit 1 as the focus, we track the individual negativity drop $\mathcal{D}_{1\setminus x}(t) = \mathcal{N}_{1|\text{rest}}(t) - \mathcal{N}_{1|\text{rest}\setminus x}(t)$ dynamically, as shown in Fig.~\ref{fig:drop_dynamics}.
For the Star graph (where qubit 1 acts as the central node), the drops with respect to any of the leaf nodes ($x \in \{2,3,4\}$) evolve identically. This uniform trajectory, corresponding to $\mathcal{D}_{1\setminus x}(|\psi_{\text{star}}\rangle)$, grows monotonically to 1, confirming that the central qubit is uniformly connected to all other nodes, such that losing any one of them impacts its global entanglement in exactly the same way.
Conversely, for the Line graph (where qubit 1 is an endpoint), the negativity drops reveal a distinct spatial hierarchy, explicitly mapping how the correlations depend on the distance from the focus qubit:
\begin{itemize}
    \item The nearest-neighbor drop $\mathcal{D}_{1\setminus 2}(|\psi_{\text{line}}\rangle)$ grows rapidly to 1, showing that qubit 1 is most strongly and directly linked to its immediate neighbor, qubit 2, which acts as the primary anchor for its entanglement throughout the process.
    \item The distant-node drop $\mathcal{D}_{1\setminus 4}(|\psi_{\text{line}}\rangle)$ remains strictly zero throughout the entire evolution. This is mathematically rigorous: the interaction $U_{34}(t)$ acts purely as a local unitary on the $1|234$ bipartition, meaning particle 4 can be traced out without ever affecting the one-to-group entanglement of particle 1.
    \item Crucially, the next-nearest-neighbor drop $\mathcal{D}_{1\setminus 3}(|\psi_{\text{line}}\rangle)$ exhibits a compelling transient behavior. It rises at intermediate times before decaying back to zero at $t = \pi/4$. This explicitly captures a transient ``correlation wave''---qubit 1 temporarily ``feels'' the entanglement contribution of qubit 3 during the evolution. The fact that the drop ultimately vanishes at the final state perfectly corroborates our earlier static cluster analysis, elegantly confirming that qubits 1 and 3 eventually decouple into distinct, disjoint inseparable clusters ($\{1,2\}$ and $\{3,4\}$).
\end{itemize}}

 {
This dynamical application confirms that our proposed negativity drop extends significantly beyond serving merely as a binary GME verifier. While conventional standard measures compress the entire state's complexity into a single scalar value, offering little insight into its internal organization, our localized structural probe dynamically tracks the evolving spatial connectivity and internal network formation of the quantum state. By reliably distinguishing topologically distinct multipartite correlations as they develop, this method provides an operationally meaningful diagnostic tool for characterizing the underlying structure of quantum many-body states and multipartite entanglement generation protocols.}

 {
\subsection{Analytical Solutions and Environmental Noise Robustness for W States}
To further validate the practical utility and scalability of our framework, we derive analytical closed-form expressions for the minimum tangle drop of $n$-qubit W states. For the $n$-qubit pure W state:
\begin{equation}
|W_n\rangle = \frac{1}{\sqrt{n}}(|10\dots0\rangle + |010\dots0\rangle + \dots + |0\dots01\rangle),
\end{equation}
the minimum tangle drop is analytically determined. Due to the permutation symmetry of the W state, the entanglement drop is invariant regardless of which particle is chosen as the focus qubit $q_1$ or which particle $i$ is traced out. By partitioning the state $|W_n\rangle$ relative to $q_1$, we obtain its Schmidt-like decomposition:
\begin{equation}
|W_n\rangle = \frac{1}{\sqrt{n}} |1\rangle_{q_1} |0\rangle^{\otimes n-1} + \sqrt{\frac{n-1}{n}} |0\rangle_{q_1} |W_{n-1}\rangle_{\text{rest}},
\end{equation}
from which the initial one-to-group tangle is calculated as $\tau_{q_1|\text{rest}} = 4(n-1)/n^2$. To evaluate the sensitivity to particle loss, we further decompose the state by isolating an arbitrary qubit $i$ ($i \neq q_1$) and the remaining $n-2$ qubits:
\begin{equation}
\begin{split}
|W_n\rangle = &\frac{1}{\sqrt{n}} |10\rangle_{q_1 i} |0\rangle^{\otimes n-2} + \frac{1}{\sqrt{n}} |01\rangle_{q_1 i} |0\rangle^{\otimes n-2} \\
&+ \sqrt{\frac{n-2}{n}} |00\rangle_{q_1 i} |W_{n-2}\rangle.
\end{split}
\label{Eq.b47}
\end{equation}
Upon tracing out qubit $i$, the system reduces to the following mixed state involving the focus qubit $q_1$ and the $(n-2)$-qubit subsystem:
\begin{equation}
\begin{split}
\rho_{q_1, \text{rest}\setminus i} = &\frac{1}{n} |0\rangle |0\rangle^{\otimes n-2} \langle 0| \langle 0|^{\otimes n-2} \\
&+ \left( \frac{1}{\sqrt{n}} |1\rangle |0\rangle^{\otimes n-2} + \sqrt{\frac{n-2}{n}} |0\rangle |W_{n-2}\rangle \right) \text{H.c.}
\end{split}
\label{Eq.c47}
\end{equation}
By treating the collective $(n-2)$ qubits as a single logical qubit with the basis $\{|0\rangle_L = |0\rangle^{\otimes n-2}, |1\rangle_L = |W_{n-2}\rangle\}$, the density matrix in  {Eq. (\ref{Eq.c47}) can be effectively mapped onto a two-qubit mixed state, whose entanglement can be directly calculated via the concurrence formula in Eq. (\ref{Eq.c4}).} For this specific reduced structure, the one-to-group tangle relative to $q_1$ is found to be $\tau_{q_1|\text{rest}\setminus i} = 4(n-2)/n^2$. By substituting these results into the definition of tangle drop $\mathcal{D}_{q_1} = \sqrt{\tau_{q_1|\text{rest}} - \tau_{q_1|\text{rest}\setminus i}}$, we obtain the closed-form expression:
\begin{equation}
\mathcal{D}_{\min} = \sqrt{\frac{4(n-1)}{n^2} - \frac{4(n-2)}{n^2}} = \frac{2}{n}.
\end{equation}
This result provides a precise quantitative baseline for evaluating multipartite entanglement in the large-$n$ limit. Notably, as the number of particles $n$ increases, the entanglement drop $\mathcal{D}_{\min}$ monotonically decreases and eventually vanishes. This indicates that W states become increasingly insensitive to the loss of a single constituent particle as the system size grows, reflecting a growing robustness of its global entanglement structure in larger networks.
}

 {
Furthermore, we investigate the robustness of this measure against a specific form of environmental noise: the mixing with the ground state $|0\dots0\rangle$. This scenario serves as a simplified model for energy dissipation or amplitude-damping-like decay to the vacuum state in various physical platforms. The resulting noisy state is given by:
\begin{equation}
\rho = (1-\gamma)|W_n\rangle\langle W_n| + \gamma |0\dots0\rangle\langle 0\dots0|,
\end{equation}
where $\gamma \in [0,1]$ denotes the noise parameter.
To evaluate the entanglement drop for this mixed state, we utilize the convex roof construction specifically optimized for rank-2 systems \cite{PhysRevA.77.032310}. Following this methodology, we consider a general pure-state superposition of the two eigenvectors:
\begin{equation}
|\Psi(\gamma, \phi)\rangle = \sqrt{1-\gamma}|W_n\rangle + e^{i\phi}\sqrt{\gamma}|0\dots0\rangle,
\end{equation}
where $\phi$ is the relative phase. By applying a derivation analogous to that of the pure W state (treating the $n-2$ qubits as a logical qubit), the minimum tangle drop for this pure superposition is determined as:
\begin{equation}
\mathcal{D}_{\min}(|\Psi(\gamma, \phi)\rangle) = \frac{2(1-\gamma)}{n}.
\end{equation}
Remarkably, this calculation reveals that the result is strictly independent of the phase $\phi$.
Furthermore, the resulting expression is found to be a linear function of the noise parameter $\gamma$. According to the fundamental property established in Ref. \cite{PhysRevA.77.032310}, this expression rigorously serves as the valid convex roof for the mixed state. Consequently, we obtain the generalized analytical solution for the noisy W state:
\begin{equation}
\mathcal{D}_{\min} = \frac{2(1-\gamma)}{n}.
\end{equation}
As illustrated in Fig.~\ref{fig4}, the minimum tangle drop $\mathcal{D}_{\min}$ exhibits a monotonic decay with respect to both the system size $n$ and the noise intensity $\gamma$. This signifies that for larger W states, the global entanglement becomes increasingly distributed, such that the loss of a single constituent particle exerts a progressively smaller relative impact on the focus qubit's one-to-group correlations. This predictable scaling behavior highlights the potential of the entanglement drop as a robust and interpretable diagnostic tool for experimental entanglement characterization in noisy near-term quantum devices.}

\begin{figure}[t]
    \centering
    \includegraphics[width=0.48\textwidth]{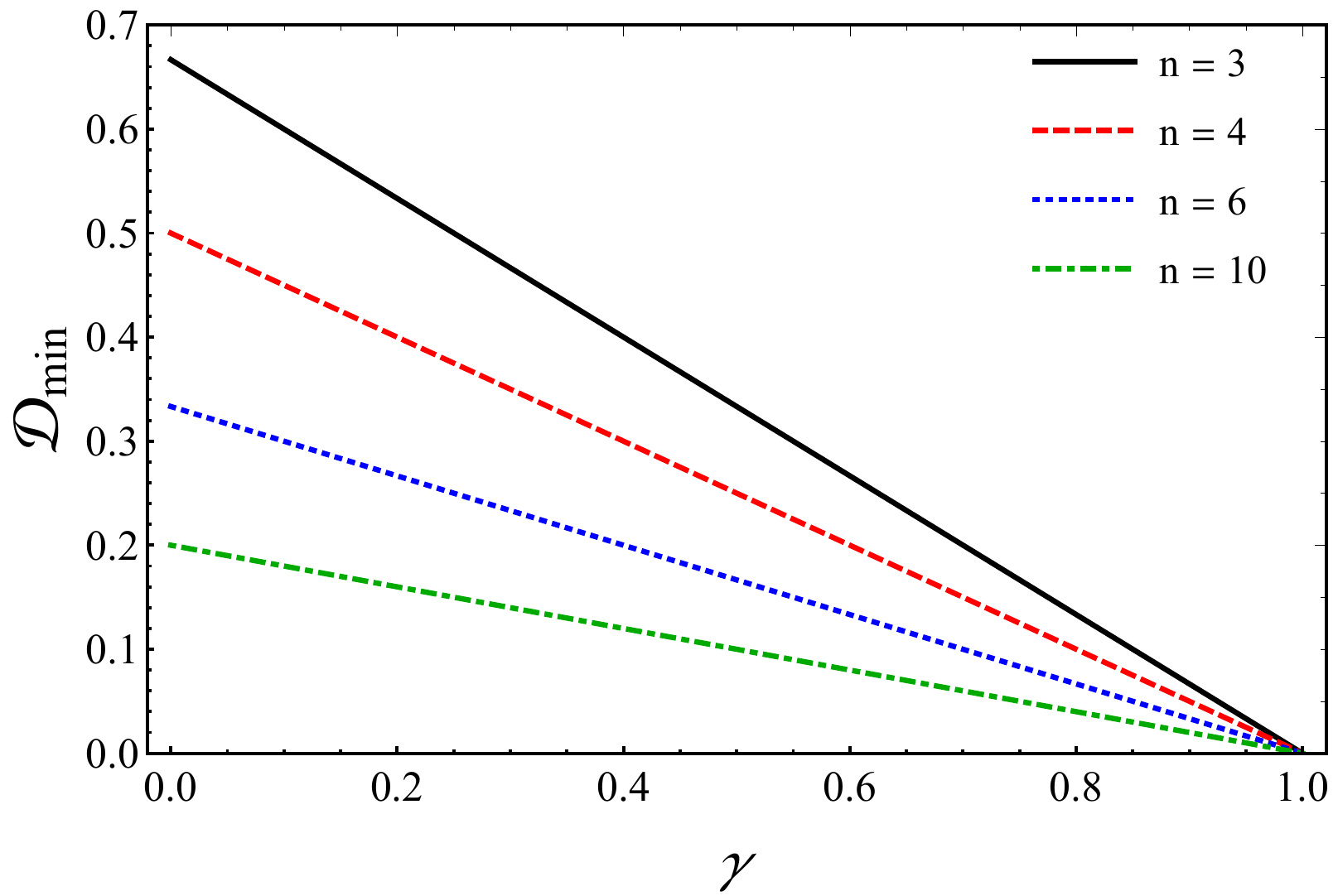}
    \caption{ {The variation of the minimum tangle drop $\mathcal{D}_{\min}$ for the $n$-qubit noisy W state as a function of the particle number $n$ and the noise parameter $\gamma$. }}
    \label{fig4}
\end{figure}

\section{CONCLUSION}

In this paper, we proposed a novel  {multipartite entanglement monotone} based on the minimum negativity drop induced by qubit loss. We rigorously proved that this quantity constitutes a valid entanglement monotone under LOCC. In the tripartite regime, we demonstrated that this measure is a faithful quantifier of GME.
 {
Beyond quantification, we established a proof-of-principle operational framework to analyze the internal connectivity of quantum systems. By evaluating the entanglement drop under particle loss, we successfully identified inseparable clusters, revealing distinct structural fingerprints that can differentiate graph topologies even within the same LC equivalence classes (e.g., the $L_1$ and $L_2$ orbits of 4-qubit graph states).  {To bridge theoretical concepts with experimental feasibility, we showed how randomized measurements via classical shadows can efficiently estimate these local negativity drops, drastically reducing  post-processing overheads. Additionally, our dynamical analysis revealed that tracking these localized drops offers a far more granular view of entanglement network evolution than conventional scalar measures.}
 {Furthermore, we established the measure's scalability in realistic scenarios by deriving exact analytical solutions for noisy $n$-qubit W states, confirming its robustness and predictable scaling against environmental degradation.}
  {However, to ensure a candid assessment of its applicability, we highlight that this structural probe is fundamentally a heuristic tool: its diagnostic capability strictly vanishes for highly robust, symmetrically correlated states, such as the 5-qubit error-correcting code.} Despite these well-defined boundaries, our framework provides a computationally efficient alternative to global optimization methods  {for modularly structured states}. By shifting the focus from exponential-scale bipartition searches to a  {localized} mapping of inseparable clusters, our approach offers a scalable diagnostic methodology
for characterizing the connectivity structure of complex
multipartite quantum systems.
}

\begin{acknowledgements}
 This work was supported by the National Science
Foundation of China under (Grants No. 12475009 and
No. 12075001), Anhui Provincial Key Research and
Development Plan (Grant No. 2022b13020004), Anhui
Province Science and Technology Innovation Project
(Grant No. 202423r06050004),  Anhui Provincial
University Scientific Research Major Project (Grant
No. 2024AH040008), Anhui Provincial Department of Industry and Information Technology (Grant No. JB24044), and the UK
Research and Innovation  (EPSRC Grant No. EP/X010929/1).
\end{acknowledgements}


%

\end{document}